\documentclass[a4paper,11pt]{elsarticle}

\usepackage{times}
\usepackage{amsmath}
\usepackage{amsthm}
\usepackage{amssymb}
\usepackage{amsfonts}
\usepackage{graphicx}
\usepackage{fullpage}
\usepackage{tikz}

\newcommand{\sketch}[1]{}

\newcommand{\tap}{{\rm TAP}}
\newcommand{\cov}{{\rm cov}}

\newcommand{\oddLP}{{\rm odd\text{-}\rm{LP}}}

\newtheorem{thm}{Theorem}[section]

\newtheorem{lemma}[thm]{Lemma}

\newtheorem{rem}[thm]{Remark}

\newcommand{\highlight}[1]{{\color{black} #1}}

\begin{document}

\begin{frontmatter}



\title{On Small-Depth Tree Augmentations}


\author[inst1]{Ojas Parekh}

\affiliation[inst1]{organization={Center for Computing Research},
            addressline={Sandia National Laboratories}, 
            city={Albuquerque},
            state={NM},
            country={USA}}

\author[inst2]{R. Ravi}
\author[inst2]{Michael Zlatin}

\affiliation[inst2]{organization={Tepper School of Business},
            addressline={Carnegie Mellon University}, 
            city={Pittsburgh},
            state={PA},
            country={USA}}

\begin{abstract}
We study the Weighted Tree Augmentation Problem for general link costs.
We show that the integrality gap of the ODD-LP relaxation for the (weighted) Tree Augmentation Problem for a $k$-level tree instance is at most $2 - \frac{1}{2^{k-1}}$.  For 2- and 3-level trees, these ratios are $\frac32$ and $\frac74$ respectively.  Our proofs are constructive and yield polynomial-time approximation algorithms with matching guarantees.
\end{abstract}


\end{frontmatter}
\title{{\bf On Small-Depth Tree Augmentations} }
\author{Ojas Parekh}



\date{\today}

\section{Introduction}
We consider the {\em weighted tree augmentation problem (TAP)}: Given an undirected graph $G=(V,E)$ with non-negative weights $c$ on the edges, and a spanning tree $T$, find a minimum cost subset of edges $A \subseteq E(G) \setminus E(T)$ such that $(V, E(T) \cup A)$ is two-edge-connected. We will call the elements of $E(T)$ as (tree) edges and those of $E(G) \setminus E(T)$ as {\em links} for convenience. A graph is \emph{two-edge-connected} if the removal of any edge does not disconnect the graph, i.e., it does not have any cut edges.
Since cut edges are also sometimes called bridges, this problem has also been called {\em bridge connectivity augmentation} in prior work~\cite{frederickson1981approximation}.

While TAP is well studied in both the weighted and unweighted case~\cite{frederickson1981approximation, khuller1993approximation, ravithesis, cohen2013, CheriyanG15a, 32KortsarzNutov, adjiashvili2016improved, fiorini2018approximating,grandoni2018improved}, it is NP-hard even when the tree has diameter $4$~\cite{frederickson1981approximation} or when the set of available links form a single cycle on the leaves of the tree $T$~\cite{cheriyan1999}, and is also APX-hard~\cite{kortsarz2004hardness}. Weighted TAP remains one of the simplest network design problems without a better than $2$-approximation in the case of general (unbounded) link costs and arbitrary depth trees, until very recently~\cite{traub2021betterthan2,traub2021local}. For the case of $n$-node trees with height $k$, Cohen and Nutov~\cite{cohen2013} gave a $(1 + \ln 2) \simeq 1.69$-approximation algorithm that runs in time $n^{3^k} \cdot poly(n)$ using an idea of Zelikovsy for approximating Steiner trees. 
Very recently, this approach has been extended to provide an approximation to the general case of the problem with the same performance guarantee by Traub and Zenklusen~\cite{traub2021betterthan2}. 
A follow-up paper by the same authors~\cite{traub2021local} improved the approximation ratio to nearly 1.5.
However, these papers do not provide any new results on the integrality gap of some natural LP relaxations for the problem that we discuss next.

\subsection{EDGE-LP Relaxation}
TAP can also be viewed as a set covering problem.
The edges of the tree $T$ define a laminar collection of cuts that are the elements to be covered using sets
represented by the links.
A link $\ell$ is said to \emph{cover} an edge $e$ if the unique cycle of $\ell+T$ contains $e$. Here we use $\cov(e)$ for a tree edge $e$ to denote the set of links which cover $e$.
The natural covering linear programming relaxation for the problem, EDGE-LP, is a special instance of a set covering problem with one requirement (element) corresponding to each cut edge in the tree. Since the tree edges define subtrees under them (after rooting it at an arbitrary node) that form a laminar family, this is also equivalent to a laminar cover problem~\cite{cheriyan1999}.
\begin{align}
\min \sum_{\ell \in E} &c_\ell x_\ell& \nonumber \\
\label{eq:edge} x(\cov(e))&\geq 1 \quad &\forall e\in E(T) \\
\label{eq:nonneg} x_\ell&\geq 0 \quad &\forall \ell\in E
\end{align}

Fredrickson and J\'aj\'a showed that the integrality gap for EDGE-LP can not exceed $2$~\cite{frederickson1981approximation} and also studied the related problem of augmenting the tree to be two-node-connected (biconnectivity versus bridge-connectivity augmentation)~\cite{fredrickson1982relationship}. Cheriyan, Jord\'an, and Ravi, who studied half-integral solutions to EDGE-LP and proved an integrality gap of $\frac43$ for such solutions, also conjectured that the overall integrality gap of EDGE-LP was at most $\frac43$~\cite{cheriyan1999}. However, Cheriyan et al.~\cite{cheriyan2008integrality} later demonstrated an instance for which the integrality gap of EDGE-LP is at least $\frac32$.

\subsection{ODD-LP Relaxation}
Fiorini et al. studied
the relaxation consisting of adding all $\{0,\frac12\}$-Chv{\'a}tal-Gomory cuts of the EDGE-LP~\cite{fiorini2018approximating}. We call their extended linear program the ODD-LP.

We define $\delta(S)$ for $S \subset V$ as the set of all links and edges with exactly one endpoint in $S$, and recall that $\cov(e)$ for a tree edge $e$ is the set of links that cover $e$.  We use $E(T)$ to refer to the set of tree edges, and $L$ is the set of links, $E(G) \setminus E(T)$.
\begin{align}
\min \sum_{\ell \in E} &c_\ell x_\ell& \nonumber \\
\label{eq:odd} x(\delta(S) \cap L)  + \sum_{e\in \delta(S)\cap E(T)} x(\cov(e)) &\geq |\delta(S)\cap E(T)|+1 \quad &\forall S \subseteq V, |\delta(S) \cap E(T)| \text{ is odd} \\
 x_\ell&\geq 0 \quad &\forall \ell\in E \nonumber
\end{align}

We describe here the validity of the constraints in ODD-LP using a proof due to Robert Carr.
Consider a set of vertices $S$ such that $|\delta(S)\cap E(T)|$ is odd.
By adding together the edge constraints for $\delta(S) \cap E(T)$ we get:
$$\sum_{e\in \delta(S)\cap E(T)} x(\cov(e)) \geq |\delta(S)\cap E(T)|$$

Now we can add any non-negative terms to the left hand side and still remain feasible. Therefore
$$x(\delta(S) \cap L)+\sum_{e\in \delta(S)\cap E(T)} x(\cov(e)) \geq |\delta(S)\cap E(T)| $$
is also feasible. Now consider any link $\ell$. If $x_\ell$ appears an even number of times in $\sum_{e\in \delta(S)\cap E(T)} x(\cov(e))$ then $\ell$ is not in $\delta(S)$. Similarly, if $x_\ell$ appears an odd number of times in $\sum_{e\in \delta(S)\cap E(T)} x(\cov(e))$ then $\ell$ is in $\delta(S)$. So, the coefficient of every $x_\ell$ on the left hand side of this expression is even. In particular, for any integer solution the left hand side is even and the right hand side is odd. Therefore, we can strengthen the right hand side by increasing it by one, and the resulting constraint will still be feasible for any integer solution. The constraint,
$$x(\delta(S) \cap L)+\sum_{e\in \delta(S)\cap E(T)} x(\cov(e)) \geq |\delta(S)\cap E(T)|+1$$
is thus valid for any integer solution to TAP as desired.

\section{Preliminaries}
We will use the following theorem about the ODD-LP~\cite{fiorini2018approximating}.
For a choice of a root $r$, we call links which connect two different components of $T - r$ as cross-links, and those that go from a node of $T$ to its ancestor as up-links.

\begin{thm}
\label{oddstar}
The ODD-LP is integral for weighted TAP instances that contain only cross- and up-links.
\end{thm}
The integrality of the formulation is shown by demonstrating that the constraint matrix is an example of a binet matrix~\cite{appa2006bidirected,appa2007optimization}, a generalization of network matrices that are a well-known class of totally unimodular matrices.
Moreover, while general Chv{\'a}tal-Gomory closures are NP-hard to optimize over, these restricted versions over half-integral combinations can be optimized in polynomial time~\cite{caprara1996}.
Such instances with only cross- and up-links are informally called ``star-shaped" with the center of the star being the chosen root, so we will refer to the above result as saying that the ODD-LP for star-shaped instances centered at a root have integrality gap 1 and solutions to such instances can be obtained in polynomial time.

Without loss of generality, we may consider TAP instances where all links go between two leaves \cite{iglesias2017coloring}. We reproduce the proof here for completeness.  

\begin{lemma}
\label{l2l}
Given an instance $(T,L,c)$ of weighted TAP, there is a corresponding, polynomial-sized instance $(T',L',c')$ with all links having both endpoints as leaves, such that there is a cost-preserving bijection between the solutions to the two instances. 
\end{lemma}


\begin{proof}
The proof proceeds by a simple graph reduction. Suppose we are given an instance defined by a graph $G$ with associated tree $T$ for the weighted TAP. We create a new instance of the leaf-to-leaf version as follows: For every internal node $u$ in the original tree $T$, we add two new leaf nodes $u'$ and $u''$ both adjacent to $u$ to get a new tree $T'$. For every link $f = (v,u)$ in the original instance, we reconnect the link to now end in the leaf $u'$ rather than the internal node $u$ in the tree $T'$. Thus, if both $v$ and $u$ are internal nodes, the new link is $(v',u')$; if only $u$ is internal, the new link is $(v,u')$ and if both are leaves, the new link is the same $(u,v)$ as in $G$. Note that the new graph $G'$ is a leaf-to-leaf instance. In addition, for every internal node $u$ in the original tree $T$, we add a new link of zero cost between $u'$ and $u''$ - this will serve to cover the newly added edges $(u,u')$ and $(u,u'')$ without changing the coverage of any of the edges in the original tree $T$. See Figure~\ref{fig:l2l}.

\begin{figure}[t]
		\centering
		\begin{tikzpicture}[scale=0.7]

		
		\begin{scope}
		
		\draw [-] [red, line width=0.4mm,xshift=0 cm] plot [smooth, tension=1] coordinates {(3,6) (0,4)};
		\draw [-] [red, line width=0.4mm,xshift=0 cm] plot [smooth, tension=1] coordinates {(3,6) (3,4)};
		\draw [-] [red, line width=0.4mm,xshift=0 cm] plot [smooth, tension=1] coordinates {(3,6) (6,4)};
		\draw [-] [red, line width=0.4mm,xshift=0 cm] plot [smooth, tension=1] coordinates {(0,4) (-1,2)};
		\draw [-] [red, line width=0.4mm,xshift=0 cm] plot [smooth, tension=1] coordinates {(0,4) (1,2)};
		\draw [-] [red, line width=0.4mm,xshift=0 cm] plot [smooth, tension=1] coordinates {(3,4) (2,2)};
		\draw [-] [red, line width=0.4mm,xshift=0 cm] plot [smooth, tension=1] coordinates {(3,4) (4,2)};
		\draw [-] [red, line width=0.4mm,xshift=0 cm] plot [smooth, tension=1] coordinates {(4,2) (5,0)};
		\draw [-] [red, line width=0.4mm,xshift=0 cm] plot [smooth, tension=1] coordinates {(4,2) (3,0)};
		
		\draw [dashed] [black, line width=0.4mm,xshift=0 cm] plot [smooth, tension=1] coordinates {(6,4) (5,0)};
		\draw [dashed] [black, line width=0.4mm,xshift=0 cm] plot [smooth, tension=1] coordinates {(6,4) (3,3.3) (0,4)};
		\draw [dashed] [black, line width=0.4mm,xshift=0 cm] plot [smooth, tension=1] coordinates {(0,4) (1.85,2.7) (4,2)};
		
		\draw[black,fill=white] (3,6) ellipse (0.15 cm  and 0.15 cm);	
		\draw[black,fill=white] (0,4) ellipse (0.15 cm  and 0.15 cm);			
		\draw[black,fill=white] (3,4) ellipse (0.15 cm  and 0.15 cm);
		\draw[black,fill=white] (6,4) ellipse (0.15 cm  and 0.15 cm);	
		\draw[black,fill=white] (-1,2) ellipse (0.15 cm  and 0.15 cm);	
		\draw[black,fill=white] (1,2) ellipse (0.15 cm  and 0.15 cm);	
		\draw[black,fill=white] (2,2) ellipse (0.15 cm  and 0.15 cm);	
		\draw[black,fill=white] (4,2) ellipse (0.15 cm  and 0.15 cm);
		\draw[black,fill=white] (3,0) ellipse (0.15 cm  and 0.15 cm);	
		\draw[black,fill=white] (5,0) ellipse (0.15 cm  and 0.15 cm);

		\node (r) at (2.6,6) {{$r$}};
		\node (u) at (-0.6,4) {{$u$}};
		\node (v) at (4.4,2) {{$v$}};
		\end{scope}
		
		\begin{scope}[xshift = 12 cm]
		
		\draw [-] [red, line width=0.4mm,xshift=0 cm] plot [smooth, tension=1] coordinates {(3,6) (0,4)};
		\draw [-] [red, line width=0.4mm,xshift=0 cm] plot [smooth, tension=1] coordinates {(3,6) (3,4)};
		\draw [-] [red, line width=0.4mm,xshift=0 cm] plot [smooth, tension=1] coordinates {(3,6) (6,4)};
		\draw [-] [red, line width=0.4mm,xshift=0 cm] plot [smooth, tension=1] coordinates {(0,4) (-2,2)};
		\draw [-] [red, line width=0.4mm,xshift=0 cm] plot [smooth, tension=1] coordinates {(0,4) (-1,2)};
		\draw [-] [red, line width=0.4mm,xshift=0 cm] plot [smooth, tension=1] coordinates {(0,4) (0,2)};
		\draw [-] [red, line width=0.4mm,xshift=0 cm] plot [smooth, tension=1] coordinates {(0,4) (1,2)};
		\draw [-] [red, line width=0.4mm,xshift=0 cm] plot [smooth, tension=1] coordinates {(3,4) (2,2)};
		\draw [-] [red, line width=0.4mm,xshift=0 cm] plot [smooth, tension=1] coordinates {(3,4) (4,2)};
		\draw [-] [red, line width=0.4mm,xshift=0 cm] plot [smooth, tension=1] coordinates {(4,2) (5,0)};
		\draw [-] [red, line width=0.4mm,xshift=0 cm] plot [smooth, tension=1] coordinates {(4,2) (4,0)};
		\draw [-] [red, line width=0.4mm,xshift=0 cm] plot [smooth, tension=1] coordinates {(4,2) (2,0)};
		\draw [-] [red, line width=0.4mm,xshift=0 cm] plot [smooth, tension=1] coordinates {(4,2) (3,0)};
		
		\draw [dashed] [black, line width=0.4mm,xshift=0 cm] plot [smooth, tension=1] coordinates {(6,4) (5,0)};
		\draw [dashed] [black, line width=0.4mm,xshift=0 cm] plot [smooth, tension=1] coordinates {(6,4) (2.5,3) (-1,2)};
		\draw [dashed] [black, line width=0.4mm,xshift=0 cm] plot [smooth, tension=1] coordinates {(-2,2) (-1.5,1.7) (-1,2)};
		\draw [dashed] [black, line width=0.4mm,xshift=0 cm] plot [smooth, tension=1] coordinates {(-1,2) (0.8,-0.6) (3,0)};
		\draw [dashed] [black, line width=0.4mm,xshift=0 cm] plot [smooth, tension=1] coordinates {(2,0) (2.5,-0.2) (3,0)};

		\draw[black,fill=white] (3,6) ellipse (0.15 cm  and 0.15 cm);	
		\draw[black,fill=white] (0,4) ellipse (0.15 cm  and 0.15 cm);			
		\draw[black,fill=white] (3,4) ellipse (0.15 cm  and 0.15 cm);
		\draw[black,fill=white] (6,4) ellipse (0.15 cm  and 0.15 cm);	
		\draw[black,fill=white] (0,2) ellipse (0.15 cm  and 0.15 cm);	
		\draw[black,fill=white] (1,2) ellipse (0.15 cm  and 0.15 cm);	
		\draw[black,fill=white] (2,2) ellipse (0.15 cm  and 0.15 cm);	
		\draw[black,fill=white] (4,2) ellipse (0.15 cm  and 0.15 cm);
		\draw[black,fill=white] (4,0) ellipse (0.15 cm  and 0.15 cm);	
		\draw[black,fill=white] (5,0) ellipse (0.15 cm  and 0.15 cm);	
		
		\draw [black,line width = 1,fill=white] (-2-0.15,2-0.15) rectangle (-2+0.15,2+0.15);
		\draw [black,line width = 1,fill=white] (-1-0.15,2-0.15) rectangle (-1+0.15,2+0.15);
		
		\draw [black,line width = 1,fill=white] (2-0.15,0-0.15) rectangle (2+0.15,0+0.15);
		\draw [black,line width = 1,fill=white] (3-0.15,0-0.15) rectangle (3+0.15,0+0.15);
		
		\node (r) at (2.6,6) {{$r$}};
		\node (u) at (-0.6,4) {{$u$}};
		\node (v) at (4.4,2) {{$v$}};
		\node (uprime) at (-1.3,1.4) {{$u'$}};
		\node (udoubleprime) at (-2.3,1.4) {{$u''$}};
		\node (vprime) at (3,-0.6) {{$v'$}};
		\node (vdoubleprime) at (1.4,0) {{$v''$}};
		\end{scope}

		\draw [->] [black, line width=1mm,xshift=0 cm] plot coordinates {(7,3)(9,3)};
		\end{tikzpicture}
		
		\caption{Transformation to a leaf to leaf instances}
		\label{fig:l2l}
\end{figure}
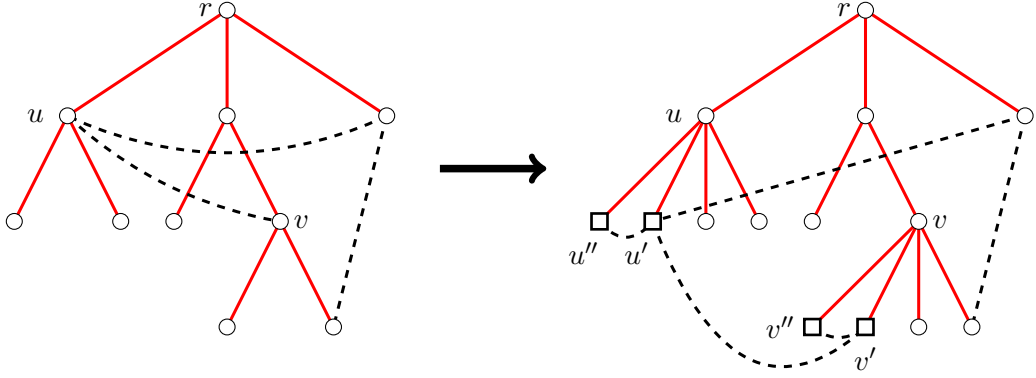


Given an solution $A \subseteq E(G) \setminus E(T)$ of minimum cost in the original instance on $G$, if we add the new zero cost edges for every internal node to $A$ we get a solution of the same cost in the new instance. Conversely, the edges in any solution $A'$ to the problem in $G'$, when restricted to the original instance is a solution of the same cost in $G$: this is because no edge of $A' \cap \cup_{\textrm{ internal nodes } u} \{ (u',u'') \}$ is useful in covering the edges of $E(T)$ in the original instance.
\end{proof}

\begin{rem}
The cost-preserving bijection described above can be extended to map fractional solutions of $\oddLP(T,L)$ to $\oddLP(T',L')$. In other words, every weighted TAP problem can be reduced to an instance where all links go between a pair of leaves without loss of generality for investigating approximation ratios and integrality gaps of the $\oddLP$.
\end{rem}

Note that given a rooted tree of $k$ levels (i.e., the maximum distance of any leaf from the root is $k$), the above transformation results in a leaf-to-leaf instance also with $k$ levels.

\section{Improved Integrality Gaps for Trees of depth 2 and 3}

\begin{thm}\label{thm:twolevel}
The integrality gap of the ODD-LP for a two-level tree instance is at most $\frac32$.
\end{thm}

\begin{proof}
First we show how to transform any integral solution $A$ into a feasible solution to two star-shaped instances, the better of which has value at most $\frac32 \cdot c(A)$. The same reduction will also apply to fractional solutions that obey the ODD-LP constraints.

We say that the root $r$ is at level 1 and its children $\{ c_1, c_2, \ldots, c_d \}$ are internal nodes at level 2, where $d$ is the number of non-leaf children of the root.
First using Lemma~\ref{l2l}, we assume that all links go between a pair of leaves.
Given an optimal solution $A$, partition the links in it into $A = A_1\dot\cup A_2$ where $A_i$ the set of links whose least common ancestor (henceforth lca) is a node in level $i$ of the tree. (Note that the lca of any link will always be an internal node in any leaf-to-leaf instance like those that we consider).

Consider now two alternate instances with feasible solutions $A'_1$ and $A'_2$ as follows.

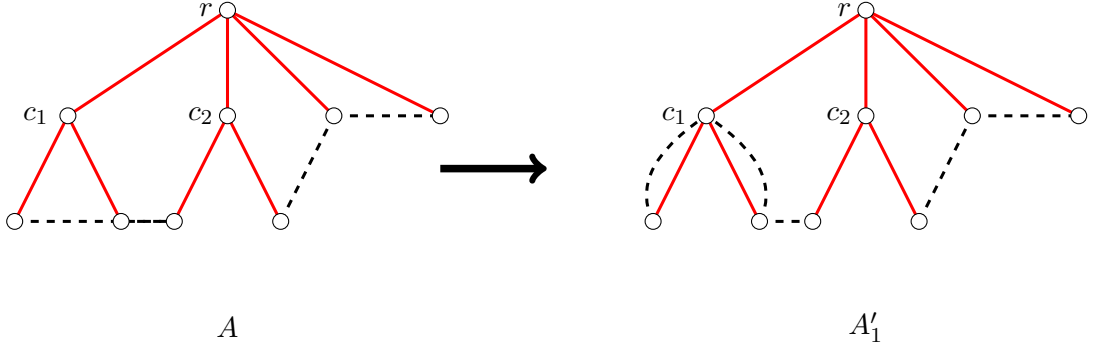
\begin{figure}[t]
		\centering
		\begin{tikzpicture}[scale=0.7]

		
		\begin{scope}
		
		\draw [-] [red, line width=0.4mm,xshift=0 cm] plot [smooth, tension=1] coordinates {(3,6) (0,4)};
		\draw [-] [red, line width=0.4mm,xshift=0 cm] plot [smooth, tension=1] coordinates {(3,6) (3,4)};
		\draw [-] [red, line width=0.4mm,xshift=0 cm] plot [smooth, tension=1] coordinates {(3,6) (5,4)};
		\draw [-] [red, line width=0.4mm,xshift=0 cm] plot [smooth, tension=1] coordinates {(3,6) (7,4)};
		\draw [-] [red, line width=0.4mm,xshift=0 cm] plot [smooth, tension=1] coordinates {(0,4) (-1,2)};
		\draw [-] [red, line width=0.4mm,xshift=0 cm] plot [smooth, tension=1] coordinates {(0,4) (1,2)};
		\draw [-] [red, line width=0.4mm,xshift=0 cm] plot [smooth, tension=1] coordinates {(3,4) (2,2)};
		\draw [-] [red, line width=0.4mm,xshift=0 cm] plot [smooth, tension=1] coordinates {(3,4) (4,2)};
		
		\draw [dashed] [black, line width=0.4mm,xshift=0 cm] plot [smooth, tension=1] coordinates {(5,4) (7,4)};
		\draw [dashed] [black, line width=0.4mm,xshift=0 cm] plot [smooth, tension=1] coordinates {(4,2) (5,4)};
		\draw [dashed] [black, line width=0.4mm,xshift=0 cm] plot [smooth, tension=1] coordinates {(1,2) (2,2)};
	    \draw [dashed] [black, line width=0.4mm,xshift=0 cm] plot [smooth, tension=1] coordinates {(-1,2) (2,2)};
		
		\draw[black,fill=white] (3,6) ellipse (0.15 cm  and 0.15 cm);	
		\draw[black,fill=white] (0,4) ellipse (0.15 cm  and 0.15 cm);			
		\draw[black,fill=white] (3,4) ellipse (0.15 cm  and 0.15 cm);
		\draw[black,fill=white] (5,4) ellipse (0.15 cm  and 0.15 cm);	
		\draw[black,fill=white] (7,4) ellipse (0.15 cm  and 0.15 cm);
		\draw[black,fill=white] (-1,2) ellipse (0.15 cm  and 0.15 cm);	
		\draw[black,fill=white] (1,2) ellipse (0.15 cm  and 0.15 cm);	
		\draw[black,fill=white] (2,2) ellipse (0.15 cm  and 0.15 cm);	
		\draw[black,fill=white] (4,2) ellipse (0.15 cm  and 0.15 cm);
		

		\node (r) at (2.6,6) {{$r$}};
		\node (c1) at (-0.6,4) {{$c_1$}};
		\node (c2) at (2.5,4) {{$c_2$}};
        \node (A) at (3,0) {$A$};

		\end{scope}
		
		\begin{scope}[xshift = 12 cm]
		
		\draw [-] [red, line width=0.4mm,xshift=0 cm] plot [smooth, tension=1] coordinates {(3,6) (0,4)};
		\draw [-] [red, line width=0.4mm,xshift=0 cm] plot [smooth, tension=1] coordinates {(3,6) (3,4)};
		\draw [-] [red, line width=0.4mm,xshift=0 cm] plot [smooth, tension=1] coordinates {(3,6) (5,4)};
		\draw [-] [red, line width=0.4mm,xshift=0 cm] plot [smooth, tension=1] coordinates {(3,6) (7,4)};
		\draw [-] [red, line width=0.4mm,xshift=0 cm] plot [smooth, tension=1] coordinates {(0,4) (-1,2)};
		\draw [-] [red, line width=0.4mm,xshift=0 cm] plot [smooth, tension=1] coordinates {(0,4) (1,2)};
		\draw [-] [red, line width=0.4mm,xshift=0 cm] plot [smooth, tension=1] coordinates {(3,4) (2,2)};
		\draw [-] [red, line width=0.4mm,xshift=0 cm] plot [smooth, tension=1] coordinates {(3,4) (4,2)};
		
		\draw [dashed] [black, line width=0.4mm,xshift=0 cm] plot [smooth, tension=1] coordinates {(5,4) (7,4)};
		\draw [dashed] [black, line width=0.4mm,xshift=0 cm] plot [smooth, tension=1] coordinates {(4,2) (5,4)};
		\draw [dashed] [black, line width=0.4mm,xshift=0 cm] plot [smooth, tension=1] coordinates {(1,2) (2,2)};

	    \draw [dashed] [black, line width=0.4mm,xshift=0 cm] plot [smooth, tension=1] coordinates {(-1,2)  (-1,3) (0,4)};
	    \draw [dashed] [black, line width=0.4mm,xshift=0 cm] plot [smooth, tension=1] coordinates {(1,2) (1,3) (0,4)};

		\draw[black,fill=white] (3,6) ellipse (0.15 cm  and 0.15 cm);	
		\draw[black,fill=white] (0,4) ellipse (0.15 cm  and 0.15 cm);			
		\draw[black,fill=white] (3,4) ellipse (0.15 cm  and 0.15 cm);
		\draw[black,fill=white] (5,4) ellipse (0.15 cm  and 0.15 cm);	
		\draw[black,fill=white] (7,4) ellipse (0.15 cm  and 0.15 cm);
		\draw[black,fill=white] (-1,2) ellipse (0.15 cm  and 0.15 cm);	
		\draw[black,fill=white] (1,2) ellipse (0.15 cm  and 0.15 cm);	
		\draw[black,fill=white] (2,2) ellipse (0.15 cm  and 0.15 cm);	
		\draw[black,fill=white] (4,2) ellipse (0.15 cm  and 0.15 cm);

		\node (r) at (2.6,6) {{$r$}};
		\node (c1) at (-0.6,4) {{$c_1$}};
		\node (c2) at (2.5,4) {{$c_2$}};
      \node (A'1) at (3,0) {$A'_1$};
		\end{scope}

		\draw [->] [black, line width=1mm,xshift=0 cm] plot coordinates {(7,3)(9,3)};
		\end{tikzpicture}
		
		\caption{Transformation to a star-shaped instance for the root}
		\label{fig:2levela}
\end{figure}

For an illustration of the first solution $A'_1$, see Figure~\ref{fig:2levela}.
For every link $(u,v)$ in $A_2$ with lca $c$ say, we replace it with two up-links $(u,c)$ and $(v,c)$ of the same cost. Note that this set of links along with $A_1$ gives a solution to a star-shaped instance centered at the root $r$. This solution has cost $c(A_1) + 2 c(A_2)$. Motivated by the existence of this solution, we can partition all the links $L = E(G) \setminus E(T)$ into $L = L_1 \dot\cup L_2$ where link $(u,v)$ is in $L_i$ if the $lca(u,v)$ is a node in level $i$ of the tree.
We then define a star-shaped instance centered at $r$ by replacing every link $(u,v)$ in $L_2$ with lca $c$ say, with two up-links $(u,c)$ and $(v,c)$ of the same cost. The minimum cost solution $A'_1$ we can find to this instance in polynomial time will have cost at most $c(A_1) + 2 c(A_2)$.

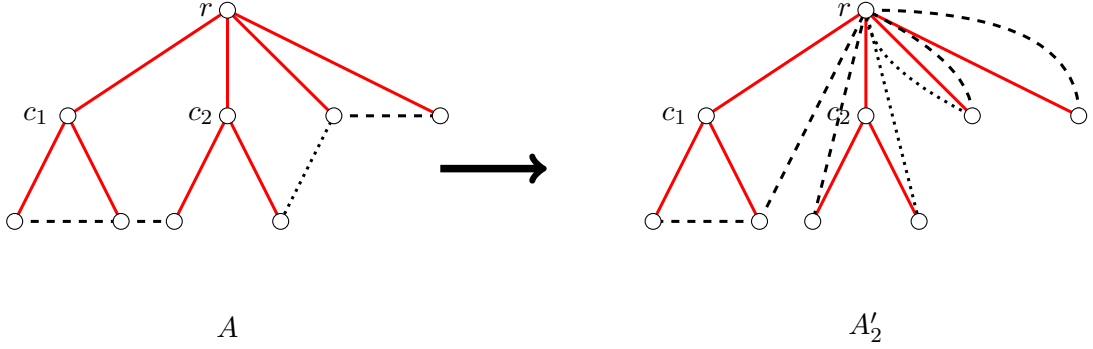
\begin{figure}[t]
		\centering
		\begin{tikzpicture}[scale=0.7]

		
		\begin{scope}
		
		\draw [-] [red, line width=0.4mm,xshift=0 cm] plot [smooth, tension=1] coordinates {(3,6) (0,4)};
		\draw [-] [red, line width=0.4mm,xshift=0 cm] plot [smooth, tension=1] coordinates {(3,6) (3,4)};
		\draw [-] [red, line width=0.4mm,xshift=0 cm] plot [smooth, tension=1] coordinates {(3,6) (5,4)};
		\draw [-] [red, line width=0.4mm,xshift=0 cm] plot [smooth, tension=1] coordinates {(3,6) (7,4)};
		\draw [-] [red, line width=0.4mm,xshift=0 cm] plot [smooth, tension=1] coordinates {(0,4) (-1,2)};
		\draw [-] [red, line width=0.4mm,xshift=0 cm] plot [smooth, tension=1] coordinates {(0,4) (1,2)};
		\draw [-] [red, line width=0.4mm,xshift=0 cm] plot [smooth, tension=1] coordinates {(3,4) (2,2)};
		\draw [-] [red, line width=0.4mm,xshift=0 cm] plot [smooth, tension=1] coordinates {(3,4) (4,2)};
		
		\draw [dashed] [black, line width=0.4mm,xshift=0 cm] plot [smooth, tension=1] coordinates {(5,4) (7,4)};
		\draw [dotted] [black, line width=0.4mm,xshift=0 cm] plot [smooth, tension=1] coordinates {(4,2) (5,4)};
		\draw [dashed] [black, line width=0.4mm,xshift=0 cm] plot [smooth, tension=1] coordinates {(1,2) (2,2)};
	    \draw [dashed] [black, line width=0.4mm,xshift=0 cm] plot [smooth, tension=1] coordinates {(-1,2) (1,2)};
		
		\draw[black,fill=white] (3,6) ellipse (0.15 cm  and 0.15 cm);	
		\draw[black,fill=white] (0,4) ellipse (0.15 cm  and 0.15 cm);			
		\draw[black,fill=white] (3,4) ellipse (0.15 cm  and 0.15 cm);
		\draw[black,fill=white] (5,4) ellipse (0.15 cm  and 0.15 cm);	
		\draw[black,fill=white] (7,4) ellipse (0.15 cm  and 0.15 cm);
		\draw[black,fill=white] (-1,2) ellipse (0.15 cm  and 0.15 cm);	
		\draw[black,fill=white] (1,2) ellipse (0.15 cm  and 0.15 cm);	
		\draw[black,fill=white] (2,2) ellipse (0.15 cm  and 0.15 cm);	
		\draw[black,fill=white] (4,2) ellipse (0.15 cm  and 0.15 cm);
		

		\node (r) at (2.6,6) {{$r$}};
		\node (c1) at (-0.6,4) {{$c_1$}};
		\node (c2) at (2.5,4) {{$c_2$}};
        \node (A) at (3,0) {$A$};

		\end{scope}
		
		\begin{scope}[xshift = 12 cm]
		
		\draw [-] [red, line width=0.4mm,xshift=0 cm] plot [smooth, tension=1] coordinates {(3,6) (0,4)};
		\draw [-] [red, line width=0.4mm,xshift=0 cm] plot [smooth, tension=1] coordinates {(3,6) (3,4)};
		\draw [-] [red, line width=0.4mm,xshift=0 cm] plot [smooth, tension=1] coordinates {(3,6) (5,4)};
		\draw [-] [red, line width=0.4mm,xshift=0 cm] plot [smooth, tension=1] coordinates {(3,6) (7,4)};
		\draw [-] [red, line width=0.4mm,xshift=0 cm] plot [smooth, tension=1] coordinates {(0,4) (-1,2)};
		\draw [-] [red, line width=0.4mm,xshift=0 cm] plot [smooth, tension=1] coordinates {(0,4) (1,2)};
		\draw [-] [red, line width=0.4mm,xshift=0 cm] plot [smooth, tension=1] coordinates {(3,4) (2,2)};
		\draw [-] [red, line width=0.4mm,xshift=0 cm] plot [smooth, tension=1] coordinates {(3,4) (4,2)};
		
		\draw [dashed] [black, line width=0.4mm,xshift=0 cm] plot [smooth, tension=1] coordinates {(5,4) (4.5,5) (3,6)};
        \draw [dashed] [black, line width=0.4mm,xshift=0 cm] plot [smooth, tension=1] coordinates {(3,6) (6, 5.5) (7,4)};
		\draw [dotted] [black, line width=0.4mm,xshift=0 cm] plot [smooth, tension=1] coordinates {(4,2) (3.5,4) (3,6)};
        \draw [dotted] [black, line width=0.4mm,xshift=0 cm] plot [smooth, tension=1] coordinates {(3,6) (3.5,5) (5,4)};
		\draw [dashed] [black, line width=0.4mm,xshift=0 cm] plot [smooth, tension=1] coordinates {(1,2) (2,4) (3,6)};
        \draw [dashed] [black, line width=0.4mm,xshift=0 cm] plot [smooth, tension=1] coordinates {(3,6) (2.5,4) (2,2)};
	    \draw [dashed] [black, line width=0.4mm,xshift=0 cm] plot [smooth, tension=1] coordinates {(1,2) (-1,2)};

		\draw[black,fill=white] (3,6) ellipse (0.15 cm  and 0.15 cm);	
		\draw[black,fill=white] (0,4) ellipse (0.15 cm  and 0.15 cm);			
		\draw[black,fill=white] (3,4) ellipse (0.15 cm  and 0.15 cm);
		\draw[black,fill=white] (5,4) ellipse (0.15 cm  and 0.15 cm);	
		\draw[black,fill=white] (7,4) ellipse (0.15 cm  and 0.15 cm);
		\draw[black,fill=white] (-1,2) ellipse (0.15 cm  and 0.15 cm);	
		\draw[black,fill=white] (1,2) ellipse (0.15 cm  and 0.15 cm);	
		\draw[black,fill=white] (2,2) ellipse (0.15 cm  and 0.15 cm);	
		\draw[black,fill=white] (4,2) ellipse (0.15 cm  and 0.15 cm);

		\node (r) at (2.6,6) {{$r$}};
		\node (c1) at (-0.6,4) {{$c_1$}};
		\node (c2) at (2.5,4) {{$c_2$}};
      \node (A'2) at (3,0) {$A'_2$};
		\end{scope}

		\draw [->] [black, line width=1mm,xshift=0 cm] plot coordinates {(7,3)(9,3)};
		\end{tikzpicture}
		
		\caption{Transformation to three star-shaped instances around the root and its two internal children}
		\label{fig:2levelb}
\end{figure}

For an illustration of the second solution $A'_2$, see Figure~\ref{fig:2levelb}.
In this case, we decompose the problem into $d+1$ different star shaped instances of which $d$ are from the subtrees defined by the star around each non-leaf child of the root, and the last is from the star defined by the root and its leaf-children. For this case, given a solution $A$ we replace every link $(u,v)$ in $A_1$ with lca the root $r$, with two up-links $(u,r)$ and $(r,v)$ of the same cost. Now consider the subtree $T_i$ defined by the star around child $c_i$ in $T$ for $i = 1,\ldots,d$. For every link in $A_1$ that has an endpoint in this subtree, one of the two copies made above goes from this endpoint to the root $r$ which is one of the leaves of this (star) tree. Similarly, the star around the root made of its leaf-children also has the copies of links in $A_1$ covering it. It is easy to verify that the subset of $A_2$ consisting of links with lca $c_i$ along with the copies of the $A_1$ links defined above give a feasible solution to the weighted TAP on this star $T_i$. This solution can be found in each such subtree as well as the star around the root,  and the sum of the costs of these solutions is $2c(A_1) + c(A_2)$. By defining the appropriate star-shaped subproblems as above, we can find in polynomial time, a solution $A'_2$ to the overall problem of cost at most $2c(A_1) + c(A_2)$.

Applying the above to the optimal solution $A^*$, we see that the best of the two solutions found above has cost at most $\min( c(A^*_1) + 2 c(A^*_2), 2c(A^*_1) + c(A^*_2) ) \leq \frac32 (c(A^*_1) + c(A^*_2)) = \frac32 c(A^*)$.

It is not hard to see that the values in any fractional solution on the links for the ODD-LP can be transformed into
a feasible fractional solution to these two sets of star shaped instances of value as claimed above. Since the resulting star shaped instances have integrality gap 1 by Theorem~\ref{oddstar}, the claim about the integrality gap also follows.
\end{proof}

\begin{thm}
The integrality gap of the ODD-LP for a three-level tree instance is at most $\frac74$.
\end{thm}

\begin{proof}
As before we will transform any integral solution $A$ into a feasible solution to one of three sets of star-shaped instances of value at most $\frac74 \cdot c(A)$. Again, the same reduction will also apply to fractional solutions that obey the ODD-LP constraints.

Using Lemma~\ref{l2l}, we assume that all links go between a pair of leaves.
Given an optimal solution $A$, partition the links in it into $A = A_1\dot\cup A_2 \dot\cup A_3$ where $A_i$ the set of links whose lca is a node in level $i$ of the tree. We say that the root $r$ is at level 1 and its non-leaf children $\{ c_1, c_2, \ldots, c_d \}$ are at level 2, and the children of these nodes that are internal nodes are in level 3 of the tree.

Consider now three alternate solutions $A'_1, A'_2$ and $A'_3$ as follows.

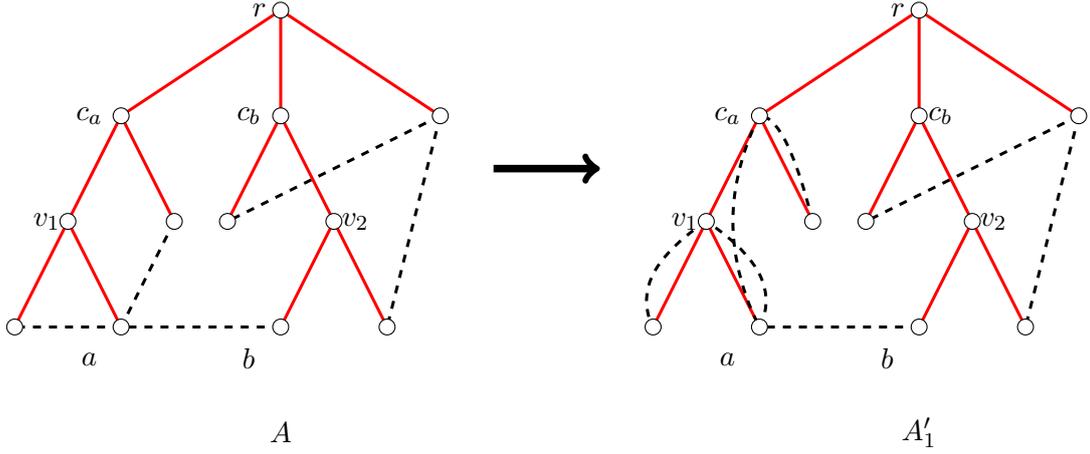
\begin{figure}[t]
		\centering
		\begin{tikzpicture}[scale=0.7]
		
		
		\begin{scope}
		
		\draw [-] [red, line width=0.4mm,xshift=0 cm] plot [smooth, tension=1] coordinates {(3,6) (0,4)};
		\draw [-] [red, line width=0.4mm,xshift=0 cm] plot [smooth, tension=1] coordinates {(3,6) (3,4)};
		\draw [-] [red, line width=0.4mm,xshift=0 cm] plot [smooth, tension=1] coordinates {(3,6) (6,4)};
		\draw [-] [red, line width=0.4mm,xshift=0 cm] plot [smooth, tension=1] coordinates {(0,4) (-1,2)};
		\draw [-] [red, line width=0.4mm,xshift=0 cm] plot [smooth, tension=1] coordinates { (-1,2) (-2,0)};		\draw [-] [red, line width=0.4mm,xshift=0 cm] plot [smooth, tension=1] coordinates {(-1,2) (0,0)};
		\draw [-] [red, line width=0.4mm,xshift=0 cm] plot [smooth, tension=1] coordinates {(0,4) (1,2)};
		\draw [-] [red, line width=0.4mm,xshift=0 cm] plot [smooth, tension=1] coordinates {(3,4) (2,2)};
		\draw [-] [red, line width=0.4mm,xshift=0 cm] plot [smooth, tension=1] coordinates {(3,4) (4,2)};
		\draw [-] [red, line width=0.4mm,xshift=0 cm] plot [smooth, tension=1] coordinates {(4,2) (5,0)};
		\draw [-] [red, line width=0.4mm,xshift=0 cm] plot [smooth, tension=1] coordinates {(4,2) (3,0)};
		
		\draw [dashed] [black, line width=0.4mm,xshift=0 cm] plot [smooth, tension=1] coordinates {(6,4) (5,0)};
		\draw [dashed] [black, line width=0.4mm,xshift=0 cm] plot [smooth, tension=1] coordinates {(3,0) (0,0)};
		\draw [dashed] [black, line width=0.4mm,xshift=0 cm] plot [smooth, tension=1] coordinates {(1,2) (0,0)};
		\draw [dashed] [black, line width=0.4mm,xshift=0 cm] plot [smooth, tension=1] coordinates {(6,4) (2,2)};
		\draw [dashed] [black, line width=0.4mm,xshift=0 cm] plot [smooth, tension=1] coordinates {(0,0) (-2,0)};
		
		\draw[black,fill=white] (3,6) ellipse (0.15 cm  and 0.15 cm);	
		\draw[black,fill=white] (0,4) ellipse (0.15 cm  and 0.15 cm);			
		\draw[black,fill=white] (3,4) ellipse (0.15 cm  and 0.15 cm);
		\draw[black,fill=white] (6,4) ellipse (0.15 cm  and 0.15 cm);	
		\draw[black,fill=white] (-1,2) ellipse (0.15 cm  and 0.15 cm);	
		\draw[black,fill=white] (1,2) ellipse (0.15 cm  and 0.15 cm);	
		\draw[black,fill=white] (2,2) ellipse (0.15 cm  and 0.15 cm);	
		\draw[black,fill=white] (4,2) ellipse (0.15 cm  and 0.15 cm);
		\draw[black,fill=white] (3,0) ellipse (0.15 cm  and 0.15 cm);	
		\draw[black,fill=white] (5,0) ellipse (0.15 cm  and 0.15 cm);		
		\draw[black,fill=white] (0,0) ellipse (0.15 cm  and 0.15 cm);			
		\draw[black,fill=white] (-2,0) ellipse (0.15 cm  and 0.15 cm);			

		\node (r) at (2.6,6) {{$r$}};
		\node (a) at (-0.6,-0.6) {{$a$}};
		\node (b) at (2.4,-0.6) {{$b$}};
		\node (ca) at (-0.6,4) {{$c_a$}};
		\node (cb) at (2.4,4) {{$c_b$}};
		\node (v2) at (4.4,2) {{$v_2$}};
		\node (v1) at (-1.4,2) {{$v_1$}};
        \node (A) at (3,-2) {$A$};
		\end{scope}
		
		\begin{scope}[xshift = 12 cm]
		
		\draw [-] [red, line width=0.4mm,xshift=0 cm] plot [smooth, tension=1] coordinates {(3,6) (0,4)};
		\draw [-] [red, line width=0.4mm,xshift=0 cm] plot [smooth, tension=1] coordinates {(3,6) (3,4)};
		\draw [-] [red, line width=0.4mm,xshift=0 cm] plot [smooth, tension=1] coordinates {(3,6) (6,4)};
		\draw [-] [red, line width=0.4mm,xshift=0 cm] plot [smooth, tension=1] coordinates {(0,4) (-1,2)};
		\draw [-] [red, line width=0.4mm,xshift=0 cm] plot [smooth, tension=1] coordinates { (-1,2) (-2,0)};		
        \draw [-] [red, line width=0.4mm,xshift=0 cm] plot [smooth, tension=1] coordinates {(-1,2) (0,0)};
		\draw [-] [red, line width=0.4mm,xshift=0 cm] plot [smooth, tension=1] coordinates {(0,4) (1,2)};
		\draw [-] [red, line width=0.4mm,xshift=0 cm] plot [smooth, tension=1] coordinates {(3,4) (2,2)};
		\draw [-] [red, line width=0.4mm,xshift=0 cm] plot [smooth, tension=1] coordinates {(3,4) (4,2)};
		\draw [-] [red, line width=0.4mm,xshift=0 cm] plot [smooth, tension=1] coordinates {(4,2) (5,0)};
		\draw [-] [red, line width=0.4mm,xshift=0 cm] plot [smooth, tension=1] coordinates {(4,2) (3,0)};
		
		\draw [dashed] [black, line width=0.4mm,xshift=0 cm] plot [smooth, tension=1] coordinates {(0,4) (-0.5,2) (0,0)};
		\draw [dashed] [black, line width=0.4mm,xshift=0 cm] plot [smooth, tension=1] coordinates {(0,0) (0,1) (-1,2)};
		\draw [dashed] [black, line width=0.4mm,xshift=0 cm] plot [smooth, tension=1] coordinates {(1,2) (0.5,3.5) (0,4)};
		\draw [dashed] [black, line width=0.4mm,xshift=0 cm] plot [smooth, tension=1] coordinates {(6,4) (5,0)};
		\draw [dashed] [black, line width=0.4mm,xshift=0 cm] plot [smooth, tension=1] coordinates {(3,0) (0,0)};
		\draw [dashed] [black, line width=0.4mm,xshift=0 cm] plot [smooth, tension=1] coordinates {(6,4) (2,2)};
		\draw [dashed] [black, line width=0.4mm,xshift=0 cm] plot [smooth, tension=1] coordinates {(-1,2) (-2,1) (-2,0)};
		\draw[black,fill=white] (3,6) ellipse (0.15 cm  and 0.15 cm);	
		\draw[black,fill=white] (0,4) ellipse (0.15 cm  and 0.15 cm);			
		\draw[black,fill=white] (3,4) ellipse (0.15 cm  and 0.15 cm);
		\draw[black,fill=white] (6,4) ellipse (0.15 cm  and 0.15 cm);	
		\draw[black,fill=white] (-1,2) ellipse (0.15 cm  and 0.15 cm);	
		\draw[black,fill=white] (1,2) ellipse (0.15 cm  and 0.15 cm);	
		\draw[black,fill=white] (2,2) ellipse (0.15 cm  and 0.15 cm);	
		\draw[black,fill=white] (4,2) ellipse (0.15 cm  and 0.15 cm);
		\draw[black,fill=white] (3,0) ellipse (0.15 cm  and 0.15 cm);	
		\draw[black,fill=white] (5,0) ellipse (0.15 cm  and 0.15 cm);		
		\draw[black,fill=white] (0,0) ellipse (0.15 cm  and 0.15 cm);			
		\draw[black,fill=white] (-2,0) ellipse (0.15 cm  and 0.15 cm);

		\node (r) at (2.6,6) {{$r$}};
		\node (a) at (-0.6,-0.6) {{$a$}};
		\node (b) at (2.4,-0.6) {{$b$}};
		\node (ca) at (-0.6,4) {{$c_a$}};
		\node (cb) at (3.4,4) {{$c_b$}};
		\node (v2) at (4.4,2) {{$v_2$}};
		\node (v1) at (-1.4,2) {{$v_1$}};
        \node (A'_1) at (3,-2) {$A'_1$};
		\end{scope}

		\draw [->] [black, line width=1mm,xshift=0 cm] plot coordinates {(7,3)(9,3)};
		\end{tikzpicture}
		
		\caption{Transformation to a star-shaped instance centered at the root}
		\label{fig:3levela}
\end{figure}

First we construct the solution $A'_1$ (See Figure~\ref{fig:3levela}) that uses links in $A_1$ once. For every link $(u,v)$ in $A_2 \cup A_3$ with lca $c$ say, we replace it with two up-links $(u,c)$ and $(v,c)$ of the same cost. Note that this set of links along with $A_1$ gives a solution to a star-shaped instance centered at the root $r$. This solution has cost $c(A_1) + 2 c(A_2) + 2c(A_3)$.
As before, to find such a solution, we can partition all the links $L = E(G) \setminus E(T)$ into $L = L_1 \dot\cup L_2 \dot\cup L_3$ where link $(u,v)$ is in $L_i$ if the $lca(u,v)$ is a node in level $i$ of the tree.
We then define a star-shaped instance centered at $r$ by replacing every link $(u,v)$ in $L_2 \cup L_3$ with lca $c$ say, with two up-links $(u,c)$ and $(v,c)$ of the same cost. The minimum cost solution $A'_1$ we can find to this instance in polynomial time will have cost at most $c(A_1) + 2 c(A_2) + 2c(A_3)$.

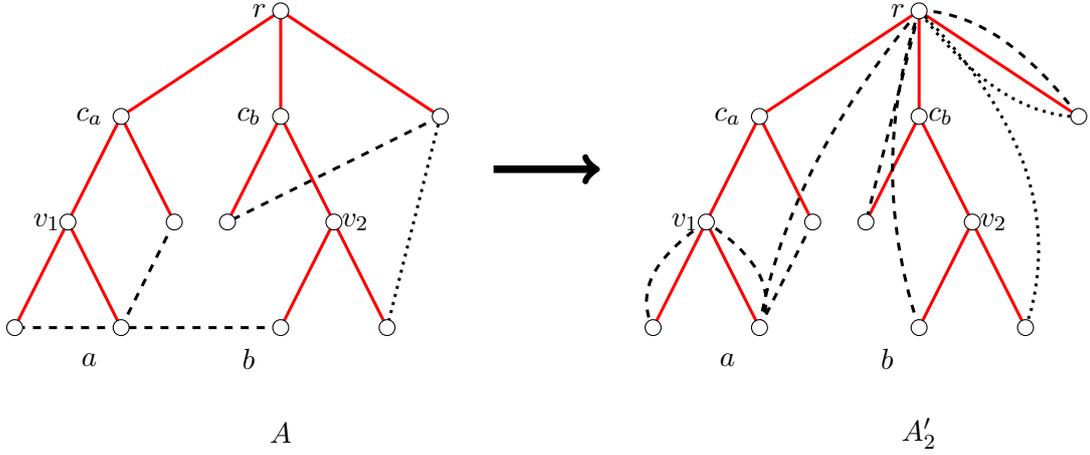
\begin{figure}[t]
		\centering
		\begin{tikzpicture}[scale=0.7]
		
		
		\begin{scope}
		
		\draw [-] [red, line width=0.4mm,xshift=0 cm] plot [smooth, tension=1] coordinates {(3,6) (0,4)};
		\draw [-] [red, line width=0.4mm,xshift=0 cm] plot [smooth, tension=1] coordinates {(3,6) (3,4)};
		\draw [-] [red, line width=0.4mm,xshift=0 cm] plot [smooth, tension=1] coordinates {(3,6) (6,4)};
		\draw [-] [red, line width=0.4mm,xshift=0 cm] plot [smooth, tension=1] coordinates {(0,4) (-1,2)};
		\draw [-] [red, line width=0.4mm,xshift=0 cm] plot [smooth, tension=1] coordinates { (-1,2) (-2,0)};		\draw [-] [red, line width=0.4mm,xshift=0 cm] plot [smooth, tension=1] coordinates {(-1,2) (0,0)};
		\draw [-] [red, line width=0.4mm,xshift=0 cm] plot [smooth, tension=1] coordinates {(0,4) (1,2)};
		\draw [-] [red, line width=0.4mm,xshift=0 cm] plot [smooth, tension=1] coordinates {(3,4) (2,2)};
		\draw [-] [red, line width=0.4mm,xshift=0 cm] plot [smooth, tension=1] coordinates {(3,4) (4,2)};
		\draw [-] [red, line width=0.4mm,xshift=0 cm] plot [smooth, tension=1] coordinates {(4,2) (5,0)};
		\draw [-] [red, line width=0.4mm,xshift=0 cm] plot [smooth, tension=1] coordinates {(4,2) (3,0)};
		
		\draw [dotted] [black, line width=0.4mm,xshift=0 cm] plot [smooth, tension=1] coordinates {(6,4) (5,0)};
		\draw [dashed] [black, line width=0.4mm,xshift=0 cm] plot [smooth, tension=1] coordinates {(3,0) (0,0)};
		\draw [dashed] [black, line width=0.4mm,xshift=0 cm] plot [smooth, tension=1] coordinates {(1,2) (0,0)};
		\draw [dashed] [black, line width=0.4mm,xshift=0 cm] plot [smooth, tension=1] coordinates {(6,4) (2,2)};
		\draw [dashed] [black, line width=0.4mm,xshift=0 cm] plot [smooth, tension=1] coordinates {(0,0) (-2,0)};
		
		\draw[black,fill=white] (3,6) ellipse (0.15 cm  and 0.15 cm);	
		\draw[black,fill=white] (0,4) ellipse (0.15 cm  and 0.15 cm);			
		\draw[black,fill=white] (3,4) ellipse (0.15 cm  and 0.15 cm);
		\draw[black,fill=white] (6,4) ellipse (0.15 cm  and 0.15 cm);	
		\draw[black,fill=white] (-1,2) ellipse (0.15 cm  and 0.15 cm);	
		\draw[black,fill=white] (1,2) ellipse (0.15 cm  and 0.15 cm);	
		\draw[black,fill=white] (2,2) ellipse (0.15 cm  and 0.15 cm);	
		\draw[black,fill=white] (4,2) ellipse (0.15 cm  and 0.15 cm);
		\draw[black,fill=white] (3,0) ellipse (0.15 cm  and 0.15 cm);	
		\draw[black,fill=white] (5,0) ellipse (0.15 cm  and 0.15 cm);		
		\draw[black,fill=white] (0,0) ellipse (0.15 cm  and 0.15 cm);			
		\draw[black,fill=white] (-2,0) ellipse (0.15 cm  and 0.15 cm);			

		\node (r) at (2.6,6) {{$r$}};
		\node (a) at (-0.6,-0.6) {{$a$}};
		\node (b) at (2.4,-0.6) {{$b$}};
		\node (ca) at (-0.6,4) {{$c_a$}};
		\node (cb) at (2.4,4) {{$c_b$}};
		\node (v2) at (4.4,2) {{$v_2$}};
		\node (v1) at (-1.4,2) {{$v_1$}};
        \node (A) at (3,-2) {$A$};
		\end{scope}
		
		\begin{scope}[xshift = 12 cm]
		
		\draw [-] [red, line width=0.4mm,xshift=0 cm] plot [smooth, tension=1] coordinates {(3,6) (0,4)};
		\draw [-] [red, line width=0.4mm,xshift=0 cm] plot [smooth, tension=1] coordinates {(3,6) (3,4)};
		\draw [-] [red, line width=0.4mm,xshift=0 cm] plot [smooth, tension=1] coordinates {(3,6) (6,4)};
		\draw [-] [red, line width=0.4mm,xshift=0 cm] plot [smooth, tension=1] coordinates {(0,4) (-1,2)};
		\draw [-] [red, line width=0.4mm,xshift=0 cm] plot [smooth, tension=1] coordinates { (-1,2) (-2,0)};		\draw [-] [red, line width=0.4mm,xshift=0 cm] plot [smooth, tension=1] coordinates {(-1,2) (0,0)};
		\draw [-] [red, line width=0.4mm,xshift=0 cm] plot [smooth, tension=1] coordinates {(0,4) (1,2)};
		\draw [-] [red, line width=0.4mm,xshift=0 cm] plot [smooth, tension=1] coordinates {(3,4) (2,2)};
		\draw [-] [red, line width=0.4mm,xshift=0 cm] plot [smooth, tension=1] coordinates {(3,4) (4,2)};
		\draw [-] [red, line width=0.4mm,xshift=0 cm] plot [smooth, tension=1] coordinates {(4,2) (5,0)};
		\draw [-] [red, line width=0.4mm,xshift=0 cm] plot [smooth, tension=1] coordinates {(4,2) (3,0)};
		
		\draw [dashed] [black, line width=0.4mm,xshift=0 cm] plot [smooth, tension=1] coordinates {(1,2) (0,0)};
		\draw [dashed] [black, line width=0.4mm,xshift=0 cm] plot [smooth, tension=1] coordinates {(3,6) (2.5, 4) (2,2)};
		\draw [dashed] [black, line width=0.4mm,xshift=0 cm] plot [smooth, tension=1] coordinates {(0,0) (0,1) (-1,2)};
		\draw [dashed] [black, line width=0.4mm,xshift=0 cm] plot [smooth, tension=1] coordinates {(6,4) (4.5, 5.5) (3,6)};
		\draw [dashed] [black, line width=0.4mm,xshift=0 cm] plot [smooth, tension=1] coordinates {(3,0) (2.5,3) (3,6)};
		\draw [dashed] [black, line width=0.4mm,xshift=0 cm] plot [smooth, tension=1] coordinates {(-1,2) (-2,1) (-2,0)};
		\draw [dashed] [black, line width=0.4mm,xshift=0 cm] plot [smooth, tension=1] coordinates {(3,6) (1,3) (0,0)};
		\draw [dotted] [black, line width=0.4mm,xshift=0 cm] plot [smooth, tension=1] coordinates {(3,6) (5,3) (5,0)};
		\draw [dotted] [black, line width=0.4mm,xshift=0 cm] plot [smooth, tension=1] coordinates {(6,4) (4.5, 4.5) (3,6)};

		\draw[black,fill=white] (3,6) ellipse (0.15 cm  and 0.15 cm);	
		\draw[black,fill=white] (0,4) ellipse (0.15 cm  and 0.15 cm);			
		\draw[black,fill=white] (3,4) ellipse (0.15 cm  and 0.15 cm);
		\draw[black,fill=white] (6,4) ellipse (0.15 cm  and 0.15 cm);	
		\draw[black,fill=white] (-1,2) ellipse (0.15 cm  and 0.15 cm);	
		\draw[black,fill=white] (1,2) ellipse (0.15 cm  and 0.15 cm);	
		\draw[black,fill=white] (2,2) ellipse (0.15 cm  and 0.15 cm);	
		\draw[black,fill=white] (4,2) ellipse (0.15 cm  and 0.15 cm);
		\draw[black,fill=white] (3,0) ellipse (0.15 cm  and 0.15 cm);	
		\draw[black,fill=white] (5,0) ellipse (0.15 cm  and 0.15 cm);		
		\draw[black,fill=white] (0,0) ellipse (0.15 cm  and 0.15 cm);			
		\draw[black,fill=white] (-2,0) ellipse (0.15 cm  and 0.15 cm);

		\node (r) at (2.6,6) {{$r$}};
		\node (a) at (-0.6,-0.6) {{$a$}};
		\node (b) at (2.4,-0.6) {{$b$}};
		\node (ca) at (-0.6,4) {{$c_a$}};
		\node (cb) at (3.4,4) {{$c_b$}};
		\node (v2) at (4.4,2) {{$v_2$}};
		\node (v1) at (-1.4,2) {{$v_1$}};
        \node (A'_2) at (3,-2) {$A'_2$};
		\end{scope}

		\draw [->] [black, line width=1mm,xshift=0 cm] plot coordinates {(7,3)(9,3)};
		\end{tikzpicture}
		
		\caption{Transformation to three star-shaped instances centered at the root and its two internal children}
		\label{fig:3levelb}
\end{figure}

For the second solution (Figure~\ref{fig:3levelb}), we proceed as before to decompose the problem into one per non-leaf neighbor $v_i$ of the root by considering the whole subtree $T_i$ under it along with its tree edge to the root, and one more for the root with its leaf children.
For this case, given a solution $A$ we replace every link $(u,v)$ in $A_1$ with lca the root $r$, with two up-links $(u,r)$ and $(r,v)$ of the same cost.
For every link $(u,v)$ in $A_3$ with lca $v'$ say, we replace it with two up-links $(u,v')$ and $(v,v')$ of the same cost.
Now consider the subtree $T_i$ defined by the non-leaf child $c_i$ in $T$ along with its tree edge to the root for $i = 1,\ldots,d$.
For every link in $A_1$ that has an endpoint in this subtree, one of the two copies made above goes from this endpoint to the root $r$ which is one of the leaves of this tree.
As before, the star around the root made of its leaf-children also has the copies of links in $A_1$ with an endpoint incident to each leaf covering the corresponding leaf child.
It is easy to verify that the solution $A_2$ consisting of links with lca $c_i$ along with the copies of the $A_1$ links defined above, and the doubled copies of links in $A_3$ give a feasible solution to the set of $d+1$ star-shaped instances of the weighted TAP on the $T_i$'s and the root. As before, a solution of at most this cost can be found in  suitably defined modified instances and the sum of the costs of these solutions is at most $2c(A_1) + c(A_2) + 2c(A_3)$.

\begin{figure}[t]
		\centering
		\begin{tikzpicture}[scale=0.7]
		
		
		\begin{scope}
		
		\draw [-] [red, line width=0.4mm,xshift=0 cm] plot [smooth, tension=1] coordinates {(3,6) (0,4)};
		\draw [-] [red, line width=0.4mm,xshift=0 cm] plot [smooth, tension=1] coordinates {(3,6) (3,4)};
		\draw [-] [red, line width=0.4mm,xshift=0 cm] plot [smooth, tension=1] coordinates {(3,6) (6,4)};
		\draw [-] [red, line width=0.4mm,xshift=0 cm] plot [smooth, tension=1] coordinates {(0,4) (-1,2)};
		\draw [-] [red, line width=0.4mm,xshift=0 cm] plot [smooth, tension=1] coordinates { (-1,2) (-2,0)};		\draw [-] [red, line width=0.4mm,xshift=0 cm] plot [smooth, tension=1] coordinates {(-1,2) (0,0)};
		\draw [-] [red, line width=0.4mm,xshift=0 cm] plot [smooth, tension=1] coordinates {(0,4) (1,2)};
		\draw [-] [red, line width=0.4mm,xshift=0 cm] plot [smooth, tension=1] coordinates {(3,4) (2,2)};
		\draw [-] [red, line width=0.4mm,xshift=0 cm] plot [smooth, tension=1] coordinates {(3,4) (4,2)};
		\draw [-] [red, line width=0.4mm,xshift=0 cm] plot [smooth, tension=1] coordinates {(4,2) (5,0)};
		\draw [-] [red, line width=0.4mm,xshift=0 cm] plot [smooth, tension=1] coordinates {(4,2) (3,0)};
		
		\draw [dotted] [black, line width=0.4mm,xshift=0 cm] plot [smooth, tension=1] coordinates {(6,4) (5,0)};
		\draw [solid] [black, line width=0.4mm,xshift=0 cm] plot [smooth, tension=1] coordinates {(3,0) (0,0)};
		\draw [dotted] [black, line width=0.4mm,xshift=0 cm] plot [smooth, tension=1] coordinates {(1,2) (0,0)};
		\draw [dashed] [black, line width=0.4mm,xshift=0 cm] plot [smooth, tension=1] coordinates {(6,4) (2,2)};
		\draw [dashed] [black, line width=0.4mm,xshift=0 cm] plot [smooth, tension=1] coordinates {(0,0) (-2,0)};
		
		\draw[black,fill=white] (3,6) ellipse (0.15 cm  and 0.15 cm);	
		\draw[black,fill=white] (0,4) ellipse (0.15 cm  and 0.15 cm);			
		\draw[black,fill=white] (3,4) ellipse (0.15 cm  and 0.15 cm);
		\draw[black,fill=white] (6,4) ellipse (0.15 cm  and 0.15 cm);	
		\draw[black,fill=white] (-1,2) ellipse (0.15 cm  and 0.15 cm);	
		\draw[black,fill=white] (1,2) ellipse (0.15 cm  and 0.15 cm);	
		\draw[black,fill=white] (2,2) ellipse (0.15 cm  and 0.15 cm);	
		\draw[black,fill=white] (4,2) ellipse (0.15 cm  and 0.15 cm);
		\draw[black,fill=white] (3,0) ellipse (0.15 cm  and 0.15 cm);	
		\draw[black,fill=white] (5,0) ellipse (0.15 cm  and 0.15 cm);		
		\draw[black,fill=white] (0,0) ellipse (0.15 cm  and 0.15 cm);			
		\draw[black,fill=white] (-2,0) ellipse (0.15 cm  and 0.15 cm);			

		\node (r) at (2.6,6) {{$r$}};
		\node (a) at (-0.6,-0.6) {{$a$}};
		\node (b) at (2.4,-0.6) {{$b$}};
		\node (ca) at (-0.6,4) {{$c_a$}};
		\node (cb) at (2.4,4) {{$c_b$}};
		\node (v2) at (4.4,2) {{$v_2$}};
		\node (v1) at (-1.4,2) {{$v_1$}};
        \node (A) at (3,-2) {$A$};
		\end{scope}
		
		\begin{scope}[xshift = 12 cm]
		
		\draw [-] [red, line width=0.4mm,xshift=0 cm] plot [smooth, tension=1] coordinates {(3,6) (0,4)};
		\draw [-] [red, line width=0.4mm,xshift=0 cm] plot [smooth, tension=1] coordinates {(3,6) (3,4)};
		\draw [-] [red, line width=0.4mm,xshift=0 cm] plot [smooth, tension=1] coordinates {(3,6) (6,4)};
		\draw [-] [red, line width=0.4mm,xshift=0 cm] plot [smooth, tension=1] coordinates {(0,4) (-1,2)};
		\draw [-] [red, line width=0.4mm,xshift=0 cm] plot [smooth, tension=1] coordinates { (-1,2) (-2,0)};		\draw [-] [red, line width=0.4mm,xshift=0 cm] plot [smooth, tension=1] coordinates {(-1,2) (0,0)};
		\draw [-] [red, line width=0.4mm,xshift=0 cm] plot [smooth, tension=1] coordinates {(0,4) (1,2)};
		\draw [-] [red, line width=0.4mm,xshift=0 cm] plot [smooth, tension=1] coordinates {(3,4) (2,2)};
		\draw [-] [red, line width=0.4mm,xshift=0 cm] plot [smooth, tension=1] coordinates {(3,4) (4,2)};
		\draw [-] [red, line width=0.4mm,xshift=0 cm] plot [smooth, tension=1] coordinates {(4,2) (5,0)};
		\draw [-] [red, line width=0.4mm,xshift=0 cm] plot [smooth, tension=1] coordinates {(4,2) (3,0)};
		
		\draw [solid] [black, line width=0.4mm,xshift=0 cm] plot [smooth, tension=1] coordinates {(0,4) (-0.5,2) (0,0)};
		\draw [dotted] [black, line width=0.4mm,xshift=0 cm] plot [smooth, tension=1] coordinates {(0,4) (0.5,2) (0,0)};
		\draw [dotted] [black, line width=0.4mm,xshift=0 cm] plot [smooth, tension=1] coordinates {(1,2) (0.5,3.5) (0,4)};
		\draw [dotted] [black, line width=0.4mm,xshift=0 cm] plot [smooth, tension=1] coordinates {(3,4) (4,3) (5,0)};
		\draw [dotted] [black, line width=0.4mm,xshift=0 cm] plot [smooth, tension=1] coordinates {(3,4) (4.5,3) (6,4)};
		\draw [dashed] [black, line width=0.4mm,xshift=0 cm] plot [smooth, tension=1] coordinates {(3,4) (4.5,5) (6,4)};
		\draw [dashed] [black, line width=0.4mm,xshift=0 cm] plot [smooth, tension=1] coordinates {(3,4) (2,3) (2,2)};
		\draw [solid] [black, line width=0.4mm,xshift=0 cm] plot [smooth, tension=1] coordinates {(3,0) (2.5,2) (3,4)};
		\draw [dashed] [black, line width=0.4mm,xshift=0 cm] plot [smooth, tension=1] coordinates {(0,0) (-2,0)};
		\draw [solid] [black, line width=0.4mm,xshift=0 cm] plot [smooth, tension=1] coordinates {(0,4) (1.5,3) (3,4)};

		\draw[black,fill=white] (3,6) ellipse (0.15 cm  and 0.15 cm);	
		\draw[black,fill=white] (0,4) ellipse (0.15 cm  and 0.15 cm);			
		\draw[black,fill=white] (3,4) ellipse (0.15 cm  and 0.15 cm);
		\draw[black,fill=white] (6,4) ellipse (0.15 cm  and 0.15 cm);	
		\draw[black,fill=white] (-1,2) ellipse (0.15 cm  and 0.15 cm);	
		\draw[black,fill=white] (1,2) ellipse (0.15 cm  and 0.15 cm);	
		\draw[black,fill=white] (2,2) ellipse (0.15 cm  and 0.15 cm);	
		\draw[black,fill=white] (4,2) ellipse (0.15 cm  and 0.15 cm);
		\draw[black,fill=white] (3,0) ellipse (0.15 cm  and 0.15 cm);	
		\draw[black,fill=white] (5,0) ellipse (0.15 cm  and 0.15 cm);		
		\draw[black,fill=white] (0,0) ellipse (0.15 cm  and 0.15 cm);			
		\draw[black,fill=white] (-2,0) ellipse (0.15 cm  and 0.15 cm);

		\node (r) at (2.6,6) {{$r$}};
		\node (a) at (-0.6,-0.6) {{$a$}};
		\node (b) at (2.4,-0.6) {{$b$}};
		\node (ca) at (-0.6,4) {{$c_a$}};
		\node (cb) at (3.4,4) {{$c_b$}};
		\node (v2) at (3.4,2) {{$v_2$}};
		\node (v1) at (-1.4,2) {{$v_1$}};
        \node (A'_3) at (3,-2) {$A'_3$};
		\end{scope}

		\draw [->] [black, line width=1mm,xshift=0 cm] plot coordinates {(7,3)(9,3)};
		\end{tikzpicture}
		
		\caption{Transformation to three star-shaped instances centered at the root and the stars around the two internal nodes in level 3}
		\label{fig:3levelc}
\end{figure}
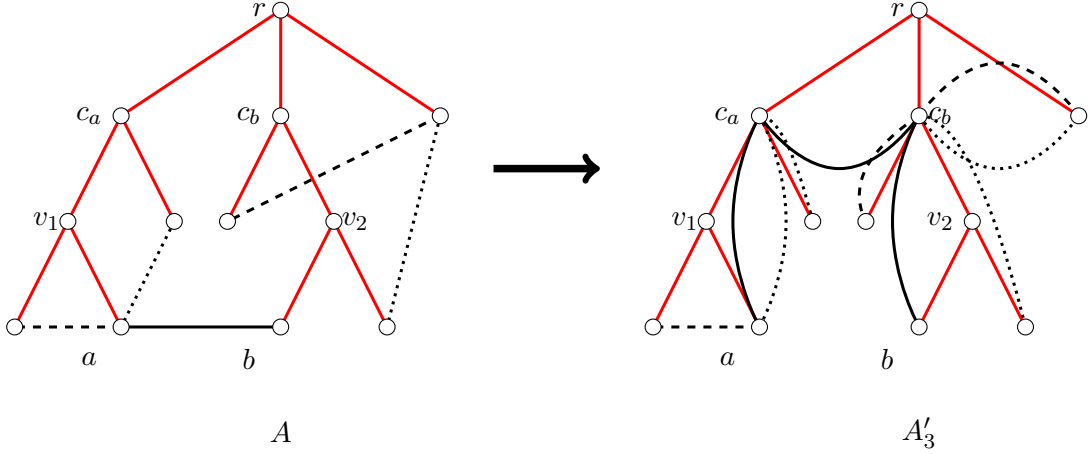

Finally, for the third solution (See Figure~\ref{fig:3levelc}), we consider the stars around the internal nodes, say $v_1,\ldots, v_q$ in level 3, and one more tree around the root consisting all the set of all tree edges not in the stars around the $v_i$'s.
To obtain a set of star-shaped solutions from $A$ for these instances we proceed as follows.
For every link $(a,b)$ in $A_2$ with lca $c$ say, we replace it with two up-links $(a,c)$ and $(b,c)$ of the same cost. Note that the lca $c$ is a leaf in one of the third level stars and so all these copies become star-shaped links for those corresponding instances.
The interesting transformation is for links in $A_1$ where we now make up to {\em three} copies.
For every link $(a,b) \in A_1$, let $c_a$ and $c_b$ denote the ancestor of $a$ and $b$ respectively in level 1. (if either $a$ or $b$ is in level 1 itself, then its ancestor in level 1 is itself). We now add three links $(a,c_a), (c_a,c_b), (c_b,b)$ of the same cost as $(a,b)$. Note that the first and third link are leaf to leaf cross links in the stars corresponding to centers $v_a$ and $v_b$ (the ancestors of $a$ and $b$ in level 2 if they exist), and that the middle link $(c_a,c_b)$ is a cross link in the star-shaped instance centered at the root.
It is now easy to verify that the copies that we have produced form a set of feasible solutions to these star-shaped instances of total cost at most $3c(A_1) + 2c(A_2) + c(A_3)$.

The best of the above three solutions corresponding to the optimal solution $A^*$ has cost at most $\min( c(A^*_1) + 2 c(A^*_2)+2c(A^*_3), 2c(A^*_1) + c(A^*_2)+2c(A^*_3), 3c(A^*_1)+2c(A^*_2)+c(A^*_3) ) \leq \frac74 (c(A^*_1) + c(A^*_2) +c(A^*_3)) = \frac74 c(A^*)$.

As before, it is not hard to see that the values in any fractional solution on the links for the ODD-LP can be transformed into
a feasible fractional solution to these three sets of star shaped instances of value as claimed above. Since the resulting star shaped instances have integrality gap 1 by Theorem~\ref{oddstar}, the claim about the integrality gap also follows.

\end{proof}

\section{Integrality gap for $k$-level trees}

With the above cases, we can now calculate an upper bound on the value of the integrality gap for general $k$-level trees where the depth of any leaf from the root is $k$.

\begin{thm}
The integrality gap of the ODD-LP for a $k$-level tree instance is at most \highlight{$2 - \frac{1}{2^{k-1}}$}.
\end{thm}

\begin{proof}
We show how to transform any integral solution $A$ into a feasible solution to one of $k$ star-shaped instances.
Partition the links in $A$ into subsets of links $A = A_1\dot\cup A_2 \ldots \dot\cup A_{k}$  where $A_l$ is the subset whose lca is a node in level $l$ of the tree for $l = 1,\ldots,k$.
Denote the cost of these subsets of links by $c_1,\ldots,c_k$ so that the total cost of $A$ is $c = \sum_{l=1}^k c_k$.
As before we set up $k$ sets of solutions, with the $l^{th}$ solution attempting to use edges in $A_l$ only once.

Note that for $l=1$, we replace all links in $A_2,\ldots,A_k$ with two links going to the lca and decompose the resulting solution into one for a star-shaped instance around the root. The cost of this candidate solution is
$$\highlight{C_1 =\ }c_1 + 2c_2 + \ldots + 2c_k.$$

\highlight{For $1 < l \leq k$}, for every internal node $v$ at level \highlight{$l$}, we consider the \highlight{subtree below it, along with the edge to its parent} and create the solution for this star-shaped instance from the solution $A$.
In addition we create one star-shaped instance around the root\highlight{, whose tree edges are disjoint from the others,} to create a final candidate solution.
First consider the \highlight{star-shaped instances} around the internal nodes $v$ in level \highlight{$l$}.
Links in $A_\highlight{l}$ are already cross links in these.
For any link $(a,b) \in \highlight{A_{l-1} \cup \bigcup_{p > l} A_p,}$ we replace it with the two links $(a,lca(a,b))$ and $(b,lca(a,b))$. \highlight{Links in $A_p$ for $p > l$ are replaced with two links that become up links in these instances.}
Consider a link $(a,b) \in A_{\highlight{l}-1}$, such that $v_a$ and $v_b$ are the \highlight{ancestors} of $a$ and $b$ respectively \highlight{that are} in level \highlight{$l$}. We replaced this link with the two links $(a,lca(a,b))$ and $(b,lca(a,b))$. Now $lca(a,b)$ is a parent of $v_a$ and $v_b$ since $(a,b) \in A_{\highlight{l}-1}$ so these links form cross links for the \highlight{star-shaped instances} around $v_a$ and $v_b$.
All the tree edges not in \highlight{any of these star-shaped instances are} considered in a final star-shaped instance \highlight{rooted at $r$}.
For links $(a,b) \in A_\highlight{q}$ for \highlight{$1<q < l-1$}, let the ancestors of $a$ and $b$ in level $\highlight{l-1}$ be $u_a$ and $u_b$ respectively\highlight{, if they exist.}
We replace $(a,b)$ with \highlight{one of the following sets, with at most four links: $\{(a, u_a),(u_a,lca(u_a,u_b)),(lca(u_a,u_b),u_b),(u_b,b)\}$, or $\{(a, lca(a,u_b)),(lca(a,u_b),u_b),(u_b,b)\}$, or $\{(a, u_a),(u_a,lca(u_a,b)),(lca(u_a,b),b)\}$, or $\{(a,lca(a,b)),(lca(a,b),b)\},$ depending on which of $u_a$ and $u_b$ exist.  For $q > 1$, all the links in these sets are cross links for the star-shaped instances around the level $l$ internal nodes or up links for the instance rooted at $r$.  Analogously, for $q=1$, we can instead use the following sets, with at most three links: $\{(a,u_a),(u_a,u_b),(u_b,b)\}$, $\{(a,u_b),(u_b,b)\}$, $\{(a,u_a),(u_a,b)\}$, $\{(a,b)\}$.  
In contrast to the $q > 1$ case, these sets also include cross links for the instance rooted at r.}

\highlight{Based on the above construction, an upper bound on} the cost of this set of candidate solutions is
\highlight{\begin{gather*}
C_2 = 2c_1 + c_2 + 2c_3 + 2c_4 + \ldots + 2c_k,\text{ if }l=2\\
C_3 = 3c_1 + 2c_2 + c_3 + 2c_4 + \ldots + 2c_k,\text{ if }l=3\\
C_l = 3c_1 + 4c_2 +\ldots + 4c_{l-2} + 2c_{l-1} + c_l + 2c_{l+1} + \ldots + 2c_1,\text{ if }l > 3.\\
\end{gather*}}
To find the worst case ratio of $\min(C_1,\ldots,C_k)$ and $c_1 + \ldots + c_k$, we show we can set the costs so that all the terms in the numerator are equal.

Setting  $C_1=C_2$ gives $c_1 + 2c_2 + \ldots + 2c_k = 2c_1 + c_2 + 2c_3 + \ldots + 2c_k$ which simplifies to
$$c_1 = c_2.$$
Setting $C_2=C_3$ gives $2c_1 + c_2 + 2c_3 + \ldots + 2c_k = 3c_1 + 2c_2 + c_3 + 2c_4 + \ldots + 2c_k$ which simplifies to
$$c_3 = 2c_1 = c_1 + c_2.$$
Setting $C_3=C_4$ gives $ 3c_1 + 2c_2 + c_3 + 2c_4 + \ldots + 2c_k = 3c_1 + \highlight{4}c_2 + 2c_3 + c_4 + 2c_5 + \ldots + 2c_k $ which simplifies to
$$c_4 = \highlight{2}c_2 + c_3.$$
In general, setting $C_{l}=C_{l+1}$ gives
$$c_{l+1} = \highlight{2}c_{l-1} + c_l.$$
The worst case ratio is then
$$\frac{C_1}{c_1 + \ldots + c_k} = \frac{c_1 + 2c_2 + \ldots + 2c_k}{c_1 + \ldots + c_k} =2 - \frac{1}{\highlight{1+\sum_{2 \leq l \leq k} 2^{l-2}}} \highlight{= 2 - \frac{1}{2^{k-1}}.}$$
\end{proof}

While the above analysis shows integrality gaps of the ODD-LP converging to 2 as the depth of the tree grows, the main open question in our opinion is to show that the integrality gap of $\frac32$ that we showed for 2-level trees is indeed the upper bound for all trees.


\section{Tight example and a lower bound on the odd-LP}

In Theorem \ref{thm:twolevel}, we showed that it is possible to obtain a feasible TAP solution of weight $c(A_1) + 2c(A_2)$, where $A = A_1 \dot \cup A_2$ is an optimal TAP solution. To improve upon the bound in Theorem \ref{thm:twolevel}, a natural idea is to try to obtain a solution of cost $c(A_1) + \alpha c(A_2)$, where $\alpha < 2$. Note that any strengthening of this form immediately yields an upper bound less than $\frac{3}{2}$ on the integrality gap of the odd-LP for 2-level TAP. However, we show that a direct improvement in this way is impossible. 

By Lemma \ref{l2l}, without loss of generality we consider a leaf-to-leaf instance $(T,L)$. If $u \in$ odd-LP$(T,L)$, we will write $u = (x,y)$ where $x$ is the projection of $u$ onto the cross-links and $y$ is its projection onto the in-links.

\begin{thm}
Let $(T,L)$ be the TAP instance given in Figure~\ref{fig:blowupinlink}, where $(x,y) = (\frac{1}{2},\frac{1}{2},\frac{1}{2},\frac{1}{2},\frac{1}{2},1)$ is an extreme point of $\oddLP(T,L)$. Then $(x,\alpha y) \not \in \tap(T,L)$ for any $\alpha < 2$. 
\end{thm}\label{thm:theexample}

\begin{proof}
Suppose $u = (x,\alpha y) \geq \frac{1}{k}\sum_{i=1}^k B_i$ where $B_i$ is the incidence vector of a integral TAP solution. Since $\ell_6$ has a value of 1, we can assume without loss of generality that link $\ell_6$ appears in $B_i$ for all $i \in [k]$. All cross-links have value $\frac{1}{2}$, so they appear in at most $\frac{k}{2}$ of the integral TAP solutions. We will show that the in-link $\ell_5$ must be used in all $k$ integral TAPs.

Note that edge $e$ is covered by exactly two cross-links $\ell_1$ and $\ell_2$. Thus, each $B_i$ must include at least one of $\ell_1$ or $\ell_2$ to be feasible. Since each of $\ell_1$ and $\ell_2$ is used in at most $\frac{k}{2}$ integral solutions, we conclude that every $B_i$ includes exactly one of $\ell_1$ and $\ell_2$. 

The minimal feasible TAP solutions which include $\ell_1$ but not $\ell_2$ are: $$\{\ell_1,\ell_6,\ell_3,\ell_4\},\{\ell_1,\ell_6,\ell_4,\ell_5\},\text{ and }\{\ell_1,\ell_6,\ell_3,\ell_5\}.$$

The minimal feasible TAP solutions which include $\ell_2$ but not $\ell_1$ are: $$\{\ell_2,\ell_6,\ell_3,\ell_4\},\{\ell_2,\ell_6,\ell_4,\ell_5\},\text{ and }\{\ell_2,\ell_6,\ell_3,\ell_5\}.$$

Thus, we may assume that each $B_i$ is one of the aforementioned feasible integral TAP solutions. Note that in all six such solutions, two links out of $\{\ell_3,\ell_4,\ell_5\}$ are used. Hence, in total, links from $\{\ell_3,\ell_4,\ell_5\}$ are used $2k$ times over all $B_i$. By assumption, the cross-links $\ell_3$ and $\ell_4$ are used at most $\frac{k}{2}$ times. Hence $\ell_5$ is used at least $k$ times. 

In particular, the in-link $\ell_5$ must be used in all $k$ integral TAP solutions in the convex combination. Since its value was only $\frac{1}{2}$, we see that $(x,\alpha y) \not \in \tap(T,L)$ for any $\alpha < 2$. \end{proof}

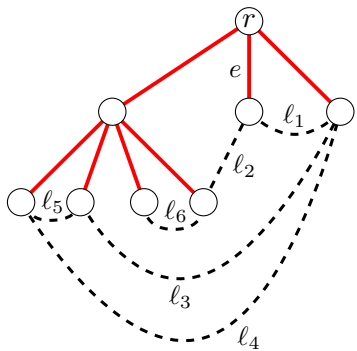
\begin{figure}[h]
		\centering
		\begin{tikzpicture}[scale=0.6]

	    \draw [-] [red, line width=0.5mm,xshift=0 cm] plot [smooth, tension=1] coordinates {(5,6) (2,4)};
		\draw [-] [red, line width=0.5mm,xshift=0 cm] plot [smooth, tension=1] coordinates {(5,6) (5,4)};
		\draw [-] [red, line width=0.5mm,xshift=0 cm] plot [smooth, tension=1] coordinates {(5,6) (7,4)};
		\draw [-] [red, line width=0.5mm,xshift=0 cm] plot [smooth, tension=1] coordinates {(0,2) (2,4)};
		\draw [-] [red, line width=0.5mm,xshift=0 cm] plot [smooth, tension=1] coordinates {(1.3,2) (2,4)};
		\draw [-] [red, line width=0.5mm,xshift=0 cm] plot [smooth, tension=1] coordinates {(2.7,2) (2,4)};
		\draw [-] [red, line width=0.5mm,xshift=0 cm] plot [smooth, tension=1] coordinates {(4,2) (2,4)};

		\draw [dashed] [black, line width=0.4mm,xshift=0 cm] plot [smooth, tension=1] coordinates {(5,4) (6,3.5) (7,4)};
		\draw [dashed] [black, line width=0.4mm,xshift=0 cm] plot [smooth, tension=2] coordinates {(0.2,2) (0.65,1.6) (1.3,2)};
		\draw [dashed] [black, line width=0.4mm,xshift=0 cm] plot [smooth, tension=2] coordinates {(2.7,2) (3.35,1.4) (4,2)};
		\draw [dashed] [black, line width=0.4mm,xshift=0 cm] plot [smooth, tension=2] coordinates {(4,2) (5,4)};
		\draw [dashed] [black, line width=0.4mm,xshift=0 cm] plot [smooth, tension=1] coordinates {(0,2) (4,-1) (7,4)};
		\draw [dashed] [black, line width=0.4mm,xshift=0 cm] plot [smooth, tension=1] coordinates {(1.3,2) (4,0.4) (7,4)};

		\draw[black,fill=white] (5,6) ellipse (0.3 cm  and 0.3 cm);	
		\draw[black,fill=white] (2,4) ellipse (0.3 cm  and 0.3 cm);			
		\draw[black,fill=white] (5,4) ellipse (0.3 cm  and 0.3 cm);
		\draw[black,fill=white] (7,4) ellipse (0.3 cm  and 0.3 cm);	
		\draw[black,fill=white] (0,2) ellipse (0.3 cm  and 0.3 cm);	
		\draw[black,fill=white] (1.3,2) ellipse (0.3 cm  and 0.3 cm);	
		\draw[black,fill=white] (2.7,2) ellipse (0.3 cm  and 0.3 cm);	
		\draw[black,fill=white] (4,2) ellipse (0.3 cm  and 0.3 cm);

		\node (r) at (5,6) {{$r$}};
		
		\node (e) at (4.7,4.9) {{\small $e$}};

		\node (l1) at (6,3.9) {{\small $\ell_1$}};
		\node (l2) at (4.9,2.8) {{\small $\ell_2$}};
		\node (l3) at (3.5,-0.1) {{\small $\ell_3$}};
		\node (l4) at (5,-1) {{\small $\ell_4$}};
		\node (l5) at (0.7,2) {{\small $\ell_5$}};
		\node (l6) at (3.4,1.8) {{\small $\ell_6$}};
		


	    
	    \end{tikzpicture}
	    \caption{In the above 2-level TAP instance $(T,L)$, the point $(x,y) = (\frac{1}{2},\frac{1}{2},\frac{1}{2},\frac{1}{2},\frac{1}{2},1)$ is an extreme point for $\oddLP(T,L)$. However, the point $(x, \alpha y)$ is not a convex combination of integral TAP solutions for any $\alpha < 2$.
}
        \label{fig:blowupinlink}
\end{figure}

\subsection{Lower Bound on the odd-LP}

Here we demonstrate a lower bound on the integrality gap of the odd-LP, even for 2-level TAP. In the TAP instance in Figure~\ref{fig:blowupinlink}, let all links have cost 1 except link $\ell_6$ which has cost 0. The optimal integral solution has cost 3. The optimal fractional solution to the odd-LP has cost at most $\frac{5}{2}$, since $(x,y) = (\frac{1}{2},\frac{1}{2},\frac{1}{2},\frac{1}{2},\frac{1}{2},1)$ is feasible. Hence the integrality gap is at least $\frac{6}{5}$. 

\section*{Acknowledgements}
Sandia National Laboratories is a multimission laboratory managed and operated by National Technology and Engineering Solutions of Sandia, LLC., a wholly owned subsidiary of Honeywell International, Inc., for the U.S. Department of Energy’s National Nuclear Security Administration under contract DE-NA-0003525. OP was supported by the U.S. Department of Energy, Office of Science, Office of Advanced Scientific Computing Research, Accelerated Research in Quantum Computing and Quantum Algorithms Teams programs.
This material is based upon work supported by the U. S. Office of Naval Research under award number N00014-21-1-2243 to RR.

\bibliographystyle{plain}
\bibliography{refs}

\end{document}